\documentclass[11pt,english]{article}
\usepackage{amsmath}
\usepackage{amstext}
\usepackage{amssymb}
\usepackage{amsthm}
\usepackage{mathptmx}
\usepackage{url}
\usepackage[authoryear]{natbib}
\usepackage{babel}
\usepackage{aliascnt}
\usepackage[includefoot,right=3.5cm,left=3.5cm,top=3.5cm,bottom=3cm]{geometry}
\newcommand{\be}{\begin{equation}}
\newcommand{\ee}{\end{equation}}
\newcommand{\bea}{\begin{eqnarray}}
\newcommand{\eea}{\end{eqnarray}}
\newcommand{\beas}{\begin{eqnarray*}}
\newcommand{\eeas}{\end{eqnarray*}}

\newtheorem{theorem}{Theorem}[section]
\newtheorem{definition}[theorem]{Definition}
\newtheorem{prop}[theorem]{Proposition}
\newtheorem{corollary}[theorem]{Corollary}

\newtheorem{example}[theorem]{Example}
\newtheorem{examples}[theorem]{Examples}
\newtheorem{foo}[theorem]{Remarks}

\newaliascnt{proposition}{theorem}
\aliascntresetthe{proposition}

\begin{document}
\title{Suitability of Capital Allocations for Performance Measurement}
\date{July 2014}
\author{Eduard Kromer\footnote{Corresponding author.}
\footnote{We would like to thank an anonymous referee for helpful remarks
and suggestions that led to an improvement of the paper.}\\
Department of Mathematics, University of Gie\ss en,\\
35392 Gie\ss en, Germany;\\
email: eduard.kromer@math.uni-giessen.de
\and
Ludger Overbeck\\
Department of Mathematics, University of Gie\ss en,\\
35392 Gie\ss en, Germany;\\
email: ludger.overbeck@math.uni-giessen.de}

\maketitle \thispagestyle{empty}
\begin{abstract}
Capital allocation principles are used in various contexts in which the
risk capital or the cost of an aggregate position has to be allocated
among its constituent parts. We study capital allocation principles
in a performance measurement framework. We introduce the notation of
suitability of allocations for performance measurement and show under
different assumptions on the involved reward and risk measures that
there exist suitable allocation methods. The existence of certain
suitable allocation principles generally is given under rather strict
assumptions on the underlying risk measure. Therefore we show, with a
reformulated definition of suitability and in a slightly modified
setting, that there is a known suitable allocation principle that
does not require any properties of the underlying risk measure.
Additionally we extend a previous characterization result from the
literature from a mean-risk to a reward-risk setting. Formulations of
this theory are also possible in a game theoretic
setting.\\[2mm]
{\bf Keywords} Risk capital allocation, subgradient allocation,
cost allocation, suitability for performance measurement, reward
measure, risk measure
\end{abstract}

\section{Introduction}

Capital allocation principles are used in various contexts in which the
risk capital or a cost of an aggregate position has to be allocated
among its constituent parts. The aggregate position, for instance, can
be a portfolio of credit derivatives, the profit, cost or turnover of a
firm or simply the outcome of a cooperative game. The importance of
capital allocation principles in credit portfolio modeling is outlined
in \cite{Kalkbrener2004}. Capital allocations can even be used to
construct tax schemes for pollution emission, see \cite{Cheridito2011}.
In any case one is interested in how to allocate (distribute) the
aggregate cost/profit among the units which are involved in generating
this cost/profit.

The main financial mathematics literature dealing with this topic
follows a structural approach to capital allocations and studies the
properties of specific allocation methods. \citet{Kalkbrener2005},
for instance, studies in detail the properties of the gradient
allocation and places them in context with the properties of the
underlying coherent risk measure. Other contributions, like
\cite{Denault2001}, \cite{Fischer2003} and \cite{McNeil2005},
similarly focus on the gradient allocation and study its properties.
This allocation method only exists if the underlying risk measure is
G{\^a}teaux-differentiable. This requirement is a strong one and
several risk measures in general fail to be
G{\^a}teaux-differentiable. For an overview of this aspect we refer
to \cite{Cheridito2011}. For the existence of the subgradient
allocation that was introduced in \cite{Delbaen2000}, we need weaker
requirements. Continuity of the underlying risk measure, to name one,
is a sufficient condition to ensure the existence of a subgradient
allocation, see Theorem 2.4.9 in \cite{Zalinescu2002}. There are also allocation principles in the literature
that do not require any properties of the risk measure and still have
some desirable structural properties. Recent examples for such
allocation principles are the ordered contribution allocations from
\cite{Cheridito2011} or the with-without allocation from
\cite{Merton1993} and \cite{Matten1996}.

In this work we will focus on another aspect of capital allocation
methods. These can be used as important tools in performance
measurement for portfolios and firms, and consequently for portfolio
optimization. The literature that studies capital allocations in this
framework, to our knowledge, narrows down to the works of
\cite{Tasche2004} and \cite{Tasche2008}. He studies the gradient
allocation principle based on G{\^a}teaux-differentiable risk
measures and uses mean-risk ratios as performance measures. He
introduces a definition of suitability for performance measurement
for capital allocations and characterizes the gradient allocation as
the unique allocation that is suitable for performance measurement.
This is a different approach to this topic compared to the structural
approaches we previously mentioned. We will take up and extend this
approach in different directions. We will show in Theorem
\ref{prop:suitability} and Theorem \ref{thm:31}, with a slightly
modified definition of suitability, that we can relax the requirement
of G{\^a}teaux-differentiability of the underlying risk measure to a
(local) continuity property. This indeed preserves the existence of
capital allocations that are suitable for performance measurement,
but we loose the uniqueness. As we will show, any subgradient
allocation is suitable for performance measurement. If the underlying
risk measure is G{\^a}teaux-differentiable, the subgradient
allocation is unique and reduces to the gradient allocation. In this
case our result reduces to the result of \cite{Tasche2004}. In
another setting, where we restrict our performance measure to
portfolios that do not allow arbitrage opportunities and
irrationality we are able to work with a stronger formulation of
suitability and to derive in this setting an existence and uniqueness
result. In the following we will work with classes of reward-risk
ratios, which were recently introduced in \cite{Cheridito2012}. These
are generalizations of mean-risk ratios from \cite{Tasche2004}. We
will show that a further modification of the definition of
suitability for performance measurement allows us to carry over the
whole approach to a cooperative game-theoretic setting and to consider
cost allocation methods. This setting has the benefit that we do not
require any properties of the underlying cost function to ensure
existence of suitable cost allocation methods. Moreover, there is a
direct connection to the risk-capital based setting that allows us to
apply these results directly to risk measures without requiring the
risk measure to have any specific properties.

The outline of the paper is as follows. In Section \ref{sec:CA} we
will introduce our notation and give a short overview of the capital
allocation methods that will be relevant for us in the remaining
work. We will work with reward-risk ratios as performance measures.
These will be introduced in Section \ref{sec:suit}. In this section
we will furthermore introduce the definition of suitability for
performance measurement with reward-risk ratios and present our main
results. In Section \ref{sec:game} we will modify the setting from
Section \ref{sec:suit} to a cooperative game-theoretic setting and show
in this new framework the existence of capital allocations that are
suitable for performance measurement.

\section{Notation, definitions and important properties of
risk measures and capital allocations}\label{sec:CA}
A financial position, the value of a firm or the outcome of a
cooperative game will be modeled as a real valued random variable $X$
from $L^p$, $1\leq p \leq \infty$, on a probability space $(\Omega,
\mathcal{F},\mathbb{P})$. As usual, we identify two random variables if
they agree $\mathbb{P}$-a.s. Inequalities between random variables are
understood in the $\mathbb{P}$-a.s.\,sense.

We will work with convex and with coherent risk measures
$\rho:L^p\rightarrow \mathbb{R}\cup\{+\infty\}$. Coherent risk measures
were introduced and studied in \cite{Artzner1999}. We will call a
functional from $L^p$ to $\mathbb{R}\cup\{+\infty\}$ a coherent risk
measure if it has the following four properties.
\begin{description}
\item[(M)] \emph{Monotonicity:} $\rho(X)\leq \rho(Y)$ for all $X,Y\in
    L^p$  such that $X\geq Y$.
\item[(T)] \emph{Translation property:} $\rho(X+m)=\rho(X)-m$ for all
    $X\in L^p$ and $m\in\mathbb{R}$.
\item[(S)] \emph{Subadditivity:} $\rho(X+Y)\leq \rho(X) + \rho(Y)$
    for all $X,Y\in L^p$.
\item[(P)] \emph{Positive homogeneity:} $\rho(\lambda X) =
    \lambda\rho(X)$ for all $X\in L^p$ and $\lambda\in \mathbb{R}_+$.
\end{description}
From the properties (S) and (P) it follows that coherent risk measures
are convex, accordingly they satisfy the property
\begin{description}
\item[(C)] \emph{Convexity:} $\rho(\lambda X + (1-\lambda)Y)\leq
    \lambda\rho(X) + (1-\lambda)\rho(Y)$ for all $X,Y\in L^p$ and
    $\lambda\in (0,1)$.
\end{description}
Functionals with the properties (M), (T) and (C) are called convex risk
measures. These were introduced and studied in \cite{Foellmer2002} and
\cite{Frittelli2002}. We will denote the domain of a convex risk
measure $\rho$ with $\text{dom\,}\rho=\{X\in L^p\,|\,\rho(X)<\infty\}$.
$\rho$ is a \emph{proper} convex risk measure if the domain of $\rho$
is nonempty. Since it only makes sense to allocate a \emph{finite} risk
capital of an aggregate position among its constituent parts we will
only work with aggregate positions $X$ from the domain of $\rho$. We
know from the extended Namioka-Klee Theorem from \cite{Frittelli2010}
that proper convex risk measures are continuous on the interior of
their domain. From this it follows that finite valued convex risk
measures on $L^p$ are continuous on $L^p$, $1\leq p \leq \infty$. In
particular proper convex risk measures on $L^{\infty}$ are
automatically finite and continuous on $L^{\infty}$. In the following
we will need the concept of the subdifferential of a convex risk
measure on $L^p$. The subdifferential of a convex risk measure at $X\in
\text{dom\,}\rho$, $1\leq p <\infty$ is the set
\begin{equation}\label{eq:subdiff}
\partial \rho(X)
=
\left\{
\left.
\xi\in L^{q}\,
\right|\,
\rho(X+Y)
\geq
\rho(X) + \mathbb{E}[\xi Y]
\quad\forall Y\in L^p
\right\},
\end{equation}
where $q$ is such that $1/q + 1/p =1$ and $L^q$ is the dual space of
$L^p$. We will call $\rho$ subdifferentiable at $X$ if
$\partial\rho(X)$ is nonempty. We know from Proposition 3.1 in
\cite{Ruszczynski2006} that $\rho$ with the properties (M) and (C) is
continuous and subdifferentiable on the interior of its domain. For
$p=\infty$ we can consider the weak*-subdifferential introduced in
\cite{Delbaen2000}
\[
\partial \rho(X)
=
\left\{
\left.
\xi\in L^1\,
\right|\,
\rho(X+Y)
\geq
\rho(X) + \mathbb{E}[\xi Y]
\quad\forall Y\in L^{\infty}
\right\}
\]
In contrast to \eqref{eq:subdiff} this set can be empty at some points
in $L^{\infty}$ even if the risk measure is finite valued and thus
continuous on $L^{\infty}$. For more details about the
weak*-subdifferential we refer to \cite{Delbaen2000}. It is known for a continuous risk measure $\rho$ that
$\partial \rho(X)$ is a singleton if and only if $\rho$ is
G\^{a}teaux-differentiable, see for instance Corollary 2.4.10 in \cite{Zalinescu2002}. Then
$\mathbb{E}(\xi Y)$ for $\xi\in\partial\rho(X)$ and $Y\in L^p$ is the
G\^{a}teaux-derivative at $X$ in the direction of $Y$, ie,
\begin{equation}\label{eq:subgradient_allocation}
  \mathbb{E}(\xi Y)
  =
  \lim_{h\rightarrow 0}
  \frac{\rho(X+h Y)-\rho(X)}{h}.
\end{equation}

All statements concerning the subdifferential of a convex risk measure
obviously remain true for the superdifferential of a concave reward
measure which, with respect to the properties (M)-(C) above, can be
understood as the negative of a risk measure. The superdifferential of
a concave reward measure $\theta:L^p \to \mathbb{R}\cup\{-\infty\}$,
$1\leq p < \infty$ at $X\in \text{dom\,}\theta$ is the set
\begin{equation}\label{eq:superdifferential}
\partial \theta(X)
=
\left\{
\left.
\psi\in L^{q}\,
\right|\,
\theta(X+Y)
\leq
\theta(X) + \mathbb{E}[\psi Y]
\quad\forall Y\in L^p
\right\},
\end{equation}
where $q$ is such that $1/q + 1/p=1$ and $L^q$ is the dual space of
$L^p$. As before, the weak*-superdifferential of $\theta$ at $X\in
\text{dom\,}\theta$ is considered for $p=\infty$ as the set
\[
\partial \theta(X)
=
\left\{
\left.
\psi\in L^1\,
\right|\,
\theta(X+Y)
\leq
\theta(X) + \mathbb{E}[\psi Y]
\quad\forall Y\in L^{\infty}
\right\}.
\]
To be precise, we will call a map $\theta:L^p\to
\mathbb{R}\cup\{-\infty\}$ a concave reward measure if it satisfies
\begin{description}
\item[($\mathbf{\overline{M}}$)] $\theta(X)\geq \theta(Y)$ for all
    $X,Y\in L^p$  such that $X\geq Y$.
\item[($\mathbf{\overline{T}}$)] $\theta(X+m)=\theta(X)+m$ for all
    $X\in L^p$ and $m\in\mathbb{R}$.
\item[($\mathbf{\overline{C}}$)] $\theta(\lambda X +
    (1-\lambda)Y)\geq \lambda\theta(X) + (1-\lambda)\theta(Y)$ for
    all $X,Y\in L^p$ and $\lambda\in (0,1)$.
\end{description}
We will call it coherent reward measure if satisfies the property (P)
in addition to the properties $(\overline{M})-(\overline{C})$.

Consider now the aggregate position $X=\sum_{i=1}^n X_i\in
\text{dom\,}\rho$, $n\in\mathbb{N}$. It is of interest for several
reasons outlined in the introduction, to allocate the aggregate risk
capital $\rho(X)$ to the $X_i$, $i=1,\ldots,n$, which are involved in
generating this risk capital. We will denote a capital allocation  as a
vector of real numbers $k=(k_1,\ldots,k_n)\in\mathbb{R}^n$. The
simplest risk capital allocation $k\in \mathbb{R}^n$ is the individual
allocation given by
\[
k_i=\rho(X_i), \ i=1,\ldots,n.
\]
If the underlying risk measure is subadditive, which is property (S) of
coherent risk measures, then the individual allocation overestimates
the total risk, ie,
\[
\rho
\left(
\sum_{i=1}^n X_i
\right)
\leq
\sum_{i=1}^n \rho(X_i)
= \sum_{i=1}^n k_i.
\]
A reasonable allocation principle should satisfy the above inequality
as equality. This is then usually called \emph{full allocation}
property or \emph{efficiency} of the allocation,
\begin{equation}
   \rho
   \left(
   \sum_{i=1}^n X_i
   \right)
   = \sum_{i=1}^n k_i.
\end{equation}
Since the individual allocation only depends on the risk of the
individual asset or subportfolio $X_i$ and does not take into account
the risk of the whole portfolio one is certainly interested in other
allocation principles. Another candidate, that additionally takes into
account the risk of the aggregate position is the with-without
allocation,
\begin{equation}\label{eq:WWA}
  k_i = \rho(X) - \rho(X-X_i).
\end{equation}
This allocation principle, known from \cite{Merton1993} and
\cite{Matten1996}, will be important in Section \ref{sec:game}. For
risk measures with property (S) both principles are related through
\[
\rho(X)-\rho(X-X_i) \leq \rho(X_i),
\]
such that the individual allocation is an upper bound for the
with-without allocation. Nevertheless, both allocation principles in
general do not satisfy the full allocation property if they are based
on coherent risk measures.

If the underlying risk measure is convex and the subdifferential of
$\rho$ at $X$ is nonempty, then for any $\xi\in
\partial \rho(X)$ we can consider the subgradient allocation,
\begin{equation}
  k_i=\mathbb{E}(\xi X_i).
\end{equation}
This allocation principle was introduced in \cite{Delbaen2000} with the
weak*-subdifferential for coherent risk measures on $L^{\infty}$. If
$\partial \rho(X)$ is a singleton the subgradient allocation reduces to
the gradient allocation which is also known as the Euler allocation,
see \cite{Tasche2008}. Then $k_i=\mathbb{E}(\xi X_i)$ is the
G\^{a}teaux-derivative at $X$ in the direction of $X_i$, ie,
\begin{equation}\label{eq:gradient_allocation_2}
  k_i
  =
  \lim_{h_i\rightarrow 0}
  \frac{\rho(X+h_i X_i)-\rho(X)}{h_i}.
\end{equation}
It is known from Euler's Theorem on positively homogeneous functions
(see \cite{Tasche2008}, Theorem A.1) that it satisfies the full
allocation property if the risk measure is
G\^{a}teaux-dif\-fe\-ren\-tiable and positively homogeneous.

In the next sections we will highlight another aspect of capital
allocations. We will study capital allocations with respect to their
ability to give an investor the right signals in a performance
measurement framework.

Before we move to the next section we will report in short the results
of \citet{Tasche2004} which form the basis for this
suitability-for-performance-measurement approach and which we will
generalize in the following. Let $X=\sum_{i=1}^n X_i$ be a portfolio
consisting of $n$ subportfolios $X_i$, $i=1,\ldots,n$. Let
$\emptyset\neq U\subset \mathbb{R}^n$, $n\in\mathbb{N}$, be an open set
and define $\rho_X:U\rightarrow \mathbb{R}$,
$\rho_X(u):=\rho(\sum_{i=1}^n u_i X_i)$. With $m_i=\mathbb{E}[X_i]$ we
similarly define $\mathbb{E}_X(u)=\mathbb{E}[\sum_{i=1}^n
u_iX_i]=\sum_{i=1}^n u_i m_i$. Here each $u\in U$ represents a
portfolio. We denote by $e_i$ the $i$-th unit vector from
$\mathbb{R}^n$, ie $e_i=(0,\ldots,0,1,0,\ldots,0)$ with $1$ at the
$i$-th position of the vector. A risk capital allocation is identified
in this setting by a map
\[
k=(k_1,\ldots,k_n):U\rightarrow\mathbb{R}^n.
\]
\citet{Tasche2004} considers in his approach mean-risk ratios, which
are also called RORAC (return on adjusted risk capital) ratios. These
are defined by $RORAC_X:U\rightarrow\mathbb{R}$,
\begin{equation}\label{eq:mean-risk-ratio}
RORAC_X(u):=\frac{\mathbb{E}_X(u)}{\rho_X(u)},
\end{equation}
where in \citet{Tasche2004} irrational portfolios ($u\in U$ such that
$\mathbb{E}_X(u)\leq 0$ and $\rho_X(u)>0$) and arbitrage portfolios
($u\in U$ such that $\mathbb{E}_X(u)>0$ and $\rho_X(u)\leq 0$) are
completely excluded from the approach. In this setting with mean-risk
ratios as  measures of performance \citet{Tasche2004} provides the
following definition of suitability of capital allocations for
performance measurement.
\begin{definition}[see \cite{Tasche2004}, Definition 3]\label{def:Tasche}
We say that a vector field $k:U\rightarrow\mathbb{R}^n$ is called suitable for
performance measurement with $\rho_X$ if it satisfies the following
conditions:
\begin{itemize}
\item[(i)] For all $m\in\mathbb{R}^n$, $u\in U$ and
    $i\in\{1,\ldots,n\}$ the inequality
\begin{equation}\label{eq:Tasche1}
m_i \rho_X(u) > k_i(u) \mathbb{E}_X(u)
\end{equation}
implies that there is an $\varepsilon>0$ such that for all
$h\in(0,\varepsilon)$ we have
\begin{equation}\label{eq:Tasche2}
RORAC_X(u+he_i)>RORAC_X(u)>RORAC_X(u-he_i).
\end{equation}
\item[(ii)] For all $m\in\mathbb{R}^n$, $u\in U$ and
    $i\in\{1,\ldots,n\}$ the inequality
\begin{equation}\label{eq:Tasche3}
m_i \rho_X(u) < k_i(u) \mathbb{E}_X(u)
\end{equation}
implies that there is an $\varepsilon>0$ such that for all
$h\in(0,\varepsilon)$ we have
\begin{equation}\label{eq:Tasche4}
RORAC_X(u+he_i)<RORAC_X(u)<RORAC_X(u-he_i).
\end{equation}
\end{itemize}
\end{definition}
Note the strong requirement that inequalities \eqref{eq:Tasche1} and
\eqref{eq:Tasche3} have to be true for all $m\in\mathbb{R}^n$.

With this definition of suitability for performance measurement
\citet{Tasche2004} provides the following existence and uniqueness
result.

\begin{theorem}[see \citet{Tasche2004}, Theorem 1]\label{thm:Tasche}
Let $\rho_X:U\rightarrow\mathbb{R}$ be a function that is partially
differentiable in $U$ with continuous derivatives. Let
$k=(k_1,\ldots,k_n):U\rightarrow\mathbb{R}^n$ be a continuous vector
field. Then the vector field $k$ is suitable for performance
measurement with $\rho_X$ if and only if
\begin{equation}
k_i(u)=\frac{\partial}{\partial u_i}\rho_X(u),\quad
i=1,\ldots,n,\, u\in U.
\end{equation}
\end{theorem}

If we take a closer look at the definition of suitability for
performance measurement from \cite{Tasche2004} we see that a certain
type of reward capital allocation for subportfolio $X_i$ is fixed,
namely $m_i=\mathbb{E}[X_i]$ in \eqref{eq:Tasche1} and
\eqref{eq:Tasche3}, whereas the risk capital allocation
$k=(k_1,\ldots,k_n)$ is in the focus of investigation. This
investigation then proves Theorem \ref{thm:Tasche}.

Our aim is to generalize the suitability framework of
\cite{Tasche2004} to a reward-risk setting with a more general
functional $\theta:L^p\to \mathbb{R}\cup\{-\infty\}$, $1\leq p \leq
\infty$, instead of the expectation. But now a problem arises since
Definition \ref{def:Tasche} can be interpreted in two ways. One way
of reading Definition \ref{def:Tasche} is to consider
\begin{equation} \label{eq:sec2:indiv}
t_i=\theta(X_i)=\mathbb{E}[X_i]
\end{equation}
to be the reward capital allocation and thus to compare the standalone
performance of $X_i$ with the performance of the portfolio
$X=\sum_{i=1}^n X_i$ by considering the ratios $\theta(X_i)/k_i$ and
$\theta(X)/\rho(X)$  in Definition \ref{def:Tasche}. Since
$\theta(X)=\mathbb{E}[X]$ is G{\^a}teaux-differentiable there arises
another interpretation of Definition \ref{def:Tasche}, namely to
consider the G{\^a}teaux-derivative of
$\theta(\cdot)=\mathbb{E}[\cdot]$ at $X$ in the direction of $X_i$,
\begin{equation}\label{eq:sec2:gat}
t_i=\lim_{h_i\to 0}\frac{\mathbb{E}[X+h_iX_i]-\mathbb{E}[X]}{h_i}
=
\mathbb{E}[X_i]
\end{equation}
As we can see in \eqref{eq:sec2:indiv} and \eqref{eq:sec2:gat}  both
concepts coincide if we consider $\theta(X)=\mathbb{E}[X]$ as a reward
measure. We will study both approaches in the following section and
show that we can generalize Tasche's result in both cases. The case
$t_i=\theta(X_i)=\mathbb{E}[X_i]$, which corresponds to the
individual reward capital allocation, will require certain structural
properties of the functional $\theta$. This will be studied at the
beginning of Section \ref{sec:suit}. The G{\^a}teaux-derivative case
will be separately studied in the Subsections \ref{subsec:3-1} and
\ref{subsec:partially_diff_rrr}, where we will provide the direct
extension of Tasche's result in Subsection
\ref{subsec:partially_diff_rrr}. In Subsection  \ref{subsec:3-1} we
won't fix any specific reward capital allocation. Instead we will
focus on sufficient properties of reward-risk allocations such that
the allocation is suitable for performance measurement.

In contrast to \citet{Tasche2004} we will consider reward-risk ratios
that were recently introduced in \cite{Cheridito2012} as measures of
performance. This leads to the next section.

\section{Suitability for performance measurement with RRRs}\label{sec:suit}
In this section we will examine the connection of reward-risk ratios to
capital allocations that are suitable for performance measurement. Let
us first introduce reward-risk ratios.

As performance measures in this work we consider reward-risk ratios
(RRRs in short) of the form
\begin{equation}\label{eq:def:reward-risk-ratio}
\alpha(X)=
\begin{cases}
0, & \mbox{if }\theta(X)\leq0\mbox{ and }\rho(X)>0,\\
+\infty, & \mbox{if }\theta(X)>0\mbox{ and }\rho(X)\leq0,\\
\frac{\theta(X)}{\rho(X)}, & \mbox{else}.
\end {cases}
\end{equation}
This means that we will measure the performance of the aggregate
position $X$ by a ratio with a risk measure $\rho: L^p \rightarrow
\mathbb{R}\cup \{+\infty\}$ in the denominator and a reward measure
$\theta: L^p \rightarrow \mathbb{R}\cup \{-\infty\}$ in the
numerator. Hereby $0/0$ and $\infty/\infty$ are understood to be $0$.
A reward measure $\theta: L^p \rightarrow \mathbb{R}\cup \{-\infty\}$
in the setting of this section is a functional with property (P) and property
\begin{description}
  \item[($\mathbf{\overline{S}}$)] \emph{Superadditivity:}
      $\theta(X+Y)\geq\theta(X)+\theta(Y)$ for all $X,Y\in L^p$,
\end{description}
Classes of reward measures with these properties were recently
studied in \citet{Cheridito2012} as robust reward measures and
distorted reward measures. We will shortly repeat their definition
here. Let $\mathcal{P}$ be a non-empty set of probability measures that
are absolutely continuous with respect to $\mathbb{P}$. The robust
reward measure is defined by
\begin{equation}\label{eq:robustR}
  \theta(X)
  =
  \inf_{Q\in \mathcal{P}}\mathbb{E}_Q[X].
\end{equation}
$\theta$ defined in \eqref{eq:robustR} obviously satisfies
($\overline{S}$) and (P). Through the set $\mathcal{P}$ we can
introduce ambiguity such that the elements of the set $\mathcal{P}$
could describe the beliefs of different traders. This allows a much
more general approach to the performance measurement framework than
simply choosing the expectation for the numerator of the reward-risk
ratio. A subclass of \eqref{eq:robustR} are distorted reward measures.
Let $\varphi:[0,1]\rightarrow [0,1]$ be a non-decreasing function
satisfying $\varphi(0)=0$ and $\varphi(1)=1$. It induces the distorted
probability $\mathbb{P}_{\varphi}[A]:=\varphi\circ\mathbb{P}[A]$.  If
$\varphi$ is convex then the Choquet integral defined by
\begin{equation}\label{eq:Choquet}
  \mathbb{E}_{\varphi}[X]
  :=
  \int_0^{\infty} \mathbb{P}_{\varphi}[X>t]dt
  +
  \int_{-\infty}^0 (\mathbb{P}_{\varphi}[X>t]-1)dt,
\end{equation}
is superadditive and positively homogeneous. Thus it satisfies
($\overline{S}$) and (P). With different choices for $\varphi$ we can
model different attitudes of agents towards rewards. For additional
information about RRRs induced by coherent or convex risk measures in
conjunction with these reward measures we refer to
\cite{Cheridito2012}. Most important for us in this work are the
properties ($\overline{S}$) and (P) and the existence of interesting
classes of reward measures with these properties. These classes
obviously contain the expectation. Thus all statements in the following
sections are also true for mean-risk ratios (see e.g.
\eqref{eq:mean-risk-ratio}) as performance measures.

Now we will examine the connection of reward-risk-ratios to capital
allocations that are suitable for performance measurement. The
following definition will clarify what we mean by suitability for
performance measurement with reward-risk ratios of type
\eqref{eq:def:reward-risk-ratio}.

\begin{definition}\label{def:suitability_G}
A risk capital allocation $k\in \mathbb{R}^n$  is suitable for
performance measurement at $X\in \text{dom\,}\rho$ with a
reward-risk-ratio $\alpha$, if for any decomposition of
$X=\sum_{i=1}^n X_i$ it satisfies the following conditions:
\begin{enumerate}
\item If $\theta(X)\geq 0$ then for any $i\in\{1,\ldots,n\}$ there
    exists $\varepsilon_i>0$ such
 that for all $h_{i}\in(0,\varepsilon_i)$
\begin{equation}\label{eq:SUIT_G1}
\frac{\theta(X_i)}{k_i}\geq \frac{\theta(X)}{\rho(X)}
\end{equation}
implies
\begin{equation}\label{eq:SUIT_G2}
\alpha(X) \geq \alpha(X-h_{i} X_{i})
\end{equation}
and
\begin{equation}\label{eq:SUIT_G3}
\frac{\theta(X_i)}{k_i}\leq \frac{\theta(X)}{\rho(X)}
\end{equation}
implies
\begin{equation}\label{eq:SUIT_G4}
\alpha(X) \geq \alpha(X+h_{i} X_{i})
\end{equation}
\item If $\theta(X)\leq 0$ then for any $i\in\{1,\ldots,n\}$ there
    exists $\varepsilon_i>0$ such that for all
    $h_{i}\in(0,\varepsilon_i)$ \eqref{eq:SUIT_G1} implies
\begin{equation}\label{eq:SUIT_G6}
 \alpha(X+h_{i} X_{i}) \geq \alpha(X)
\end{equation}%
and \eqref{eq:SUIT_G3} implies
\begin{equation}\label{eq:SUIT_G8}
 \alpha(X- h_{i} X_{i}) \geq \alpha(X)
\end{equation}%
\end{enumerate}%
\end{definition}

The motivation behind this definition of suitability for performance
measurement is similar to the one in \cite{Tasche2004} and
\cite{Tasche2008}. Evidently, \eqref{eq:SUIT_G1} and \eqref{eq:SUIT_G3}
relate the performance of the aggregate position $X$, measured by
$\alpha(X)$, to the performance of each unit $X_i$, $i=1,\ldots,n$,
which is involved in generating the risk capital $\rho(X)$. The
performance of each unit $X_i$, $i=1,\ldots,n$, is measured by
$\theta(X_i)/k_i$ (assume for the sake of simplicity that $k_i>0$ and
$\rho(X)>0$) and only those capital allocations $k=(k_1,\ldots,k_n)$
are suitable for performance measurement that send the right signals to
the agent. The right signals in this context are represented by
\eqref{eq:SUIT_G2} and \eqref{eq:SUIT_G4} (in connection to
\eqref{eq:SUIT_G1} and \eqref{eq:SUIT_G3}) and imply not to reduce the
capital of position $X_i$ if the standalone performance of $X_i$ is
better than the performance of the aggregate position $X$. On the other
hand a worse standalone performance of $X_i$ should send the signal not
to invest additional capital in this position. Summarized, this means
\begin{itemize}
\item 
$\alpha(X)\geq \alpha(X-h_iX_i)$ do not reduce the capital in
position $X_i$

\item 
$\alpha(X)\geq\alpha(X+h_iX_i)$ do not increase the capital in
position $X_i$

\item 
$\alpha(X+h_iX_i)\geq \alpha(X)$ increase the capital in  position
$X_i$

\item 
$\alpha(X-h_iX_i)\geq \alpha(X)$ reduce the capital in  position
$X_i$
\end{itemize}

We will see in Subsection \ref{subsec:partially_diff_rrr} that there is
a stronger formulation of suitability for performance measurement, but
it also requires much stronger additional properties of the underlying
risk measure to guarantee the existence of suitable capital
allocations.

For further statements we will need the following technical assumption
on the underlying risk measure.
\begin{description}
\item[(A)]

For an aggregate position $X\in \text{dom\,}\rho$ and any
decomposition $(X_1,\ldots,X_n)$ of $X=\sum_{i=1}^n X_i$ there
exists $\varepsilon_i>0$ such that for all
$h_i\in(0,\varepsilon_i)$ we have
\begin{itemize}
\item
if $\rho(X)>0$ then $\rho(X-h_iX_i)>0$ and $\rho(X+h_iX_i)>0$,
\item
if $\rho(X)<0$ then $\rho(X-h_iX_i)<0$ and $\rho(X+h_iX_i)<0$.
\end{itemize}
for any $i\in \{1,\ldots,n\}$.
\end{description}
It should be noted that convex risk measures that are continuous at
$X\in\text{dom\,}\rho$, automatically satisfy assumption (A) at $X$.

Now we can examine which properties a capital allocation should have
to be suitable for performance measurement with a reward-risk ratio
$\alpha$ at $X\in \text{dom\,}\rho$.
\begin{theorem}\label{prop:suitability}
Let $\theta:L^p\rightarrow\mathbb{R}\cup \{-\infty\}$, $1\leq p \leq
\infty$, be a functional with the properties ($\bar{S}$) and (P). Let $\rho:L^p\rightarrow
\mathbb{R}\cup\{+\infty\}$, $1\leq p \leq \infty$, be a map that
satisfies condition (A) in $X\in \text{dom\,}\rho$ and let $\alpha$
be a reward risk ratio of type \eqref{eq:def:reward-risk-ratio}. Then
a risk capital allocation $k\in \mathbb{R}^n$ that satisfies the
properties
\begin{align}
h_i k_i
& \leq \rho(X+h_i X_i)-\rho(X)\quad \text{and} \label{eq:prop_alloc_1}\\
h_i k_i
& \geq \rho(X) - \rho(X-h_i X_i)\label{eq:prop_alloc_2}
\end{align}
for all $h_i\in(0,\varepsilon_i)$, $\varepsilon_i$ from (A), $1\leq i
\leq n$, and for any decomposition $(X_1,\ldots,X_n)$ of
$X=\sum_{i=1}^n X_i$ is suitable for performance measurement with
$\alpha$ at $X$.
\end{theorem}
\begin{proof}
Consider first the case where numerator and denominator have the same
signs. Let $\theta(X)\geq 0$. With assumption (A) for $\rho$ at
$X\in\text{dom\,}\rho$ there exists $\varepsilon_i>0$ such that for
all $h_i\in (0,\varepsilon_i)$ we have $\rho(X)\rho(X+h_i X_i)>0$ and
$\rho(X)\rho(X-h_i X_i)>0$ for any $i\in\{1,\ldots,n\}$. From
($\overline{S}$) and (P) we get $\theta(X-h_iX_i)\leq
\theta(X)-h_i\theta(X_i)$. This, together with \eqref{eq:SUIT_G1} and
\eqref{eq:prop_alloc_2} leads to
\begin{align*}
 & \theta(X-h_{i}X_{i})\rho(X)-\theta(X)\rho(X-h_{i}X_{i})\\
\leq & \theta(X)(\rho(X)-h_i k_i-\rho(X-h_{i}X_{i}))\leq 0,
\end{align*}
which immediately implies $\alpha(X)\geq\alpha(X-h_iX_i)$. By
analogous steps we get from \eqref{eq:SUIT_G3} and
\eqref{eq:prop_alloc_1} the property $\alpha(X)\geq\alpha(X+h_iX_i)$.

Now let $\theta(X)\leq 0$. Then again with assumption (A) for $\rho$
in $X\in\text{dom\,}\rho$ there exists $\varepsilon_i>0$ such that
for all $h_i\in (0,\varepsilon_i)$ we have $\rho(X)\rho(X+h_i X_i)>0$
and $\rho(X)\rho(X-h_i X_i)>0$ for any $i\in\{1,\ldots,n\}$. From
($\overline{S}$), (P), \eqref{eq:SUIT_G1} and \eqref{eq:prop_alloc_1}
we get
\begin{align*}
 & \theta(X+h_{i}X_{i})\rho(X)-\theta(X)\rho(X+h_{i}X_{i})\\
\geq & \theta(X)(\rho(X)+h_i k_i-\rho(X+h_{i}X_{i}))\geq 0,
\end{align*}
which leads to $\alpha(X)\leq\alpha(X+h_iX_i)$. By analogous steps we
get from \eqref{eq:SUIT_G3} and \eqref{eq:prop_alloc_2} the property
$\alpha(X)\leq\alpha(X-h_iX_i)$.

Let us now consider the case with $\theta(X)\leq 0$ and $\rho(X)>0$.
Here we have $\alpha(X)=0$ and from
\eqref{eq:SUIT_G1} it follows that
\[
\alpha(X+h_i X_i)
\geq
\alpha(X) \frac{\rho(X)+h_i k_i}{\rho(X+h_i X_i)}=0.
\]
 From \eqref{eq:SUIT_G3} it
follows that
\[
\alpha(X-h_i X_i)
\geq
\alpha (X)
\frac{\rho(X)-h_i k_i}{\rho(X-h_i X_i)}
=0.
\]

The remaining relevant case, in which $\theta(X)>0$ and $\rho(X)<0$,
leads to $\alpha(X)=+\infty$  and since $\alpha(X+h_i X_i)\leq +\infty$
and $\alpha(X-h_i X_i)\leq + \infty$ we have finished the proof.
\end{proof}

In Theorem \ref{prop:suitability}, except for condition (A), we didn't
require any further properties like coherency or convexity from the
risk measure $\rho$. We have only used the superadditivity and the
positive homogeneity of $\theta$, which are the properties
($\overline{S}$) and (P). The result nevertheless does not guarantee
the existence of allocations that satisfy the required properties of
this theorem. However it gives us a suitability-criterion, which is
easy to verify and thus gives us the opportunity to check whether a
capital allocation is suitable for performance measurement or not,
without requiring the risk measure to have any properties except (A).
For instance we can immediately see that the with-without allocation
does not satisfy both of the properties \eqref{eq:prop_alloc_1} and
\eqref{eq:prop_alloc_2}. For the following corollary that tells us
which capital allocations satisfy \eqref{eq:prop_alloc_1} and
\eqref{eq:prop_alloc_2} and thus are suitable for performance
measurement  we will need the nonemptiness of the subdifferential of
$\rho$ at $X$.

\begin{corollary}\label{cor:subgr}
Let $\theta:L^p \rightarrow \mathbb{R}\cup\{-\infty\}$, $1\leq p \leq
\infty$, be a be a functional with the properties ($\bar{S}$) and (P). Let $\rho:L^p\rightarrow
\mathbb{R}\cup\{+\infty\}$, $1\leq p \leq \infty$, be a convex risk
measure with nonempty subdifferential at $X\in \text{dom\,}\rho$. Let
$\alpha$ be a reward-risk ratio of type
\eqref{eq:def:reward-risk-ratio}. If (A) is true for $\rho$ at $X \in
\text{dom\,}\rho$ then for any $\xi\in\partial \rho(X)$ the
subgradient risk capital allocation $k\in \mathbb{R}^n$ defined by
\begin{equation}\label{eq:SGA}
k_i:=\mathbb{E}[\xi X_i],
\quad \text{for any}\quad i\in\{1,\ldots,n\}
\end{equation}
is suitable for performance measurement with $\alpha$ at $X$.
\end{corollary}
\begin{proof}
Since $\partial \rho(X)$ is nonempty there exits $\xi \in \partial
\rho(X)$ for $X\in \text{dom\,}\rho\subseteq L^p$, $1\leq p \leq
\infty$ such that $k_i=\mathbb{E}[\xi X_i]$, $i=1,\ldots,n$,
obviously satisfies \eqref{eq:prop_alloc_1} for any decomposition of
$X$. Furthermore for any decomposition $(X_1,\ldots,X_n)$ of
$X=\sum_{i=1}^n X_i$ it satisfies
\[
\rho(X-h_i X_i) - \rho(X)
\geq
\mathbb{E}[\xi (-h_iX_i)]			
\]
for any $h_i\in (0,\varepsilon_i)$, $i=1,\ldots,n$. This is
equivalent to \eqref{eq:prop_alloc_2}. Thus with Theorem
\ref{prop:suitability} for any $\xi\in\partial \rho(X)$ any
subgradient capital allocation $k_i=\mathbb{E}[\xi X_i]$,
$\xi\in\partial\rho(X)$, $i=1,\ldots,n$, is suitable for performance
measurement with $\alpha$ at $X$.
\end{proof}

As outlined in Section \ref{sec:CA} we know that convex risk measures
on $L^p$, $1\leq p <\infty$, are continuous and subdifferentiable on
the interior of their domain. Thus if the aggregate position $X$ is
from the interior of the domain of $\rho$ condition (A) and the
nonemptiness of the subdifferential of $\rho$ at $X$ are
automatically satisfied. This leads to the next corollary.

\begin{corollary}
Let $\theta:L^p \rightarrow \mathbb{R}\cup\{-\infty\}$, $1\leq p <\infty$ be a functional with the properties ($\bar{S}$) and (P). Let $\rho:L^p\rightarrow \mathbb{R}\cup\{+\infty\}$, $1\leq p <\infty$ be a
convex risk measure and let $\alpha$ be a reward-risk ratio of type
\eqref{eq:def:reward-risk-ratio}. If the aggregate position
$X=\sum_{i=1}^n X_i$  is from the interior
of the domain of $\rho$ then the subdifferential of $\rho$ at $X$ is
nonempty and any subgradient risk capital allocation defined in
\eqref{eq:SGA} is suitable for performance measurement with $\alpha$
at $X$.
\end{corollary}

If we directly assume that the underlying convex risk measure is
continuous at the aggregate position $X$, then assumption (A) is
satisfied, the subdifferential of $\rho$ at $X$ is nonempty and thus
the subgradient allocation exists and is suitable for performance
measurement with a reward-risk ratio $\alpha$. Since the continuity
of a convex risk measure $\rho$ on $L^p$ follows from the finiteness
of $\rho$ on $L^p$, $1\leq p \leq \infty$ we can formulate the
following corollary.

\begin{corollary}
Let $\theta:L^p \rightarrow \mathbb{R}$, $1\leq p < \infty$ be a functional with the properties ($\bar{S}$) and (P). Let
$\rho:L^p\rightarrow \mathbb{R}$, $1\leq p < \infty$ be a convex risk measure and let
$\alpha$ be a reward-risk ratio of type
\eqref{eq:def:reward-risk-ratio}. Then the subdifferential of $\rho$
at $X$ is nonempty for any $X\in L^p$ and any
subgradient risk capital allocation is suitable for performance
measurement with $\alpha$ at any $X\in L^p$.
\end{corollary}
As already mentioned in Section \ref{sec:CA},
G\^ateaux-differentiability of a continuous $\rho$ at $X$ leads  to the uniqueness of the subgradient of $\rho$ at
$X$. In this case the G\^ateaux-derivative is the only element of the
subdifferential of $\rho$ at $X$ and thus the capital allocation
defined in \eqref{eq:gradient_allocation_2} is suitable for performance
measurement with a reward-risk ratio $\alpha$ at $X$. This allows us to
formulate the final corollary from Theorem \ref{prop:suitability}.
\begin{corollary}\label{cor:gradient-}
Let $\theta:L^p \rightarrow \mathbb{R}\cup\{-\infty\}$, $1\leq p\leq \infty$ be a functional with the properties ($\bar{S}$) and (P). Let $\alpha$ be a reward-risk ratio of type
\eqref{eq:def:reward-risk-ratio}. If the risk measure $\rho :L^p
\rightarrow \mathbb{R}\cup\{+\infty\}$, $1\leq p\leq \infty$ is continuous and  G{\^a}teaux-differentiable
at the aggregate position $X$, then the
gradient risk capital allocation, defined by
\begin{equation}
  k_i
  =
  \lim_{h_i\rightarrow 0}
  \frac{\rho(X+h_i X_i)-\rho(X)}{h_i},
  \quad
  i=1,\ldots,n,
\end{equation}
exists and is suitable for performance measurement with $\alpha$ at
$X$.
\end{corollary}
\begin{proof}
This follows directly from Corollary \ref{cor:subgr} and the reasoning
from Section \ref{sec:CA}.
\end{proof}

\begin{example}\label{ex:distortion}
Assume that the underlying probability space is non-atomic.
Distortion-ex\-po\-nen\-tial risk measures from \citet{Tsanakas2003}
are defined as
\begin{equation}
\rho_{\psi,a}(X)=\frac{1}{a}\ln\mathbb{E}_{\psi}(\exp(-aX)),
\end{equation}
where $a> 0$ and $\mathbb{E}_{\psi}[\cdot]$ is the Choquet integral
from \eqref{eq:Choquet} with a concave map $\psi$. Let $\psi$ be
differentiable. Then we know from Proposition 2 in \citet{Tsanakas2009}
that $\rho_{\psi,a}$ is G{\^a}teaux-differentiable at $X$, if and
only if $F_X^{-1}$ is strictly increasing. The gradient risk capital allocation in
this case is given by
\[
k_i=\frac{\mathbb{E}[-X_ie^{-aX}\psi'(F_X(X))]}
{\mathbb{E}[e^{-aX}\psi'(F_X(X))]},
\]
where $F_X^{-1}$ is the quantile function
\[
F_X^{-1}(p)=\inf\{x\in\mathbb{R}:F_X(x)\geq p\}.
\]
Note that with $\psi(t)=t$ this class of risk measures covers the
entropic risk measure.

If we consider distortion risk measures that are defined as
Choquet-integrals with a concave map $\psi$, see \eqref{eq:Choquet},
then from \citet{Carlier2003} the following result holds. Let $\psi$
be differentiable. Then $\rho_{\psi}(X)=\mathbb{E}_{\psi}[-X]$ is
G{\^a}teaux-differentiable if and only if $F_X^{-1}$ is strictly
increasing. In this case the gradient risk capital allocation is
given by
\[
k_i=\mathbb{E}[-X_i\psi'(F_X(X))].
\]
\end{example}

As we have mentioned in Section \ref{sec:CA} the interpretation $t_i=\theta(X_i) = \mathbb{E}[X_i]$ of the reward capital allocation in Definition \ref{def:Tasche} is not the only one possible. Thus we will study the other case
\[
t_i
=\lim_{h_i\to 0}\frac{\mathbb{E}[X+h_iX_i]-\mathbb{E}[X]}{h_i}
=\mathbb{E}[X_i]
\]
in the following two subsections.  In Subsection
\ref{subsec:reward-risk-allocations} we won't fix a specific reward
capital allocation in the definition of suitability for performance
measurement. In case of G{\^a}teaux-differentiable reward and risk
measures this will nevertheless lead to gradient reward and risk
capital allocations. In Subsection \ref{subsec:partially_diff_rrr} we
will follow a similar approach as \cite{Tasche2004} (see Definition
\ref{def:Tasche}), fix the gradient reward capital allocation in the
numerator and show that the only risk capital allocation that is
suitable for performance measurement is the gradient risk capital
allocation.

\subsection{Suitability of reward-risk allocations for
performance
measurement}\label{subsec:reward-risk-allocations}\label{subsec:3-1}
One way of dealing with portfolios that lead to \emph{arbitrage
opportunities} ($\theta(X)>0$ and $\rho(X)\leq 0$) or are
\emph{irrational} ($\theta(X)\leq 0$, $\rho(X)>0$) is to set the
performance measure for such portfolios to $+\infty$ respectively
$0$. This is the way we have followed in the previous section.
Another way is to restrict the performance measure to portfolios that
do not allow for such possibilities. This is the direction we will
follow in this subsection. Let $\mathcal{X}$ be a subset of $L^p$,
$1\leq p \leq \infty$, such that for each $X\in\mathcal{X}$ either
\begin{equation}
\theta(X)>0\quad\text{and}\quad\rho(X)>0
\end{equation}
or
\begin{equation}
\theta(X)<0\quad\text{and}\quad\rho(X)<0
\end{equation}

Here we do not fix any specific reward capital allocation in the definition of suitability for performance measurement as we have done in Definition \ref{def:suitability_G}. Instead we will modify the definition of suitability to incorporate an arbitrary reward capital allocation $t=(t_1,\ldots,t_n)$ and show in the following results the existence of reward-risk capital allocations $(t,k)$ that are suitable for performance measurement.

\begin{definition}\label{def:31_1}
A reward-risk allocation $(t,k)\in\mathbb{R}^n\times\mathbb{R}^n$ is
suitable for performance measurement at
$X\in\mathcal{X}\cap\text{dom\,}\rho\cap\text{dom\,}\theta$ with a
reward-risk ratio $\alpha$ if for any decomposition of
$X=\sum_{i=1}^n X_i$ it satisfies the following conditions:
\begin{enumerate}
\item

If $\theta(X)>0$ then for any $i\in\{1,\ldots,n\}$ there exists
$\varepsilon_i>0$ such that for all $h_i\in(0,\varepsilon_i)$
\begin{equation}\label{eq:31_1}
\frac{t_i}{k_i}>\frac{\theta(X)}{\rho(X)}
\end{equation}
implies
\begin{equation}\label{eq:31_2}
\alpha(X)>\alpha(X-h_iX_i)
\end{equation}
and
\begin{equation}\label{eq:31_3}
\frac{t_i}{k_i}<\frac{\theta(X)}{\rho(X)}
\end{equation}
implies
\begin{equation}\label{eq:31_4}
\alpha(X)>\alpha(X+h_iX_i).
\end{equation}

\item

If $\theta(X)<0$ then for any $i\in\{1,\ldots,n\}$ there exists
$\varepsilon_i>0$ such that for all $h_i\in(0,\varepsilon_i)$
\eqref{eq:31_1} implies
\begin{equation}
\alpha(X+h_iX_i)>\alpha(X)
\end{equation}
and \eqref{eq:31_3} implies
\begin{equation}
\alpha(X-h_iX_i)>\alpha(X).
\end{equation}
\end{enumerate}
\end{definition}

The interpretation of this definition of suitability for performance
measurement remains the same as for Definition
\ref{def:suitability_G}. We have only changed the measure of the
standalone performance of the units $X_i$, $i=1,\ldots,n$. In
Definition \ref{def:suitability_G} we measure the standalone
performance of $X_i$ by $\theta(X_i)/k_i$ whereas in Definition
\ref{def:31_1} we change this to the ratio of reward-risk allocations
$t_i/k_i$, $i=1,\ldots,n$. Again it is assumption (A) that allows us
to formulate and prove a modification of Theorem
\ref{prop:suitability} to the suitability setting we have introduced
in Definition \ref{def:31_1}.

\begin{theorem}\label{thm:31}
Let $\rho:L^p\to\mathbb{R}\cup\{+\infty\}$, $1\leq p \leq \infty$, be
a map that satisfies condition (A) for $X\in\mathcal{X} \cap
\text{dom\,}\rho \cap \text{dom\,}\theta$. Then any reward-risk
allocation $(t,k)\in\mathbb{R}^n\times\mathbb{R}^n$ that satisfies
the properties
\begin{align}
\label{eq:31_5}h_it_i \geq \theta(X+h_iX_i)-\theta(X),\\
\label{eq:31_6}h_it_i \leq \theta(X)-\theta(X-h_iX_i),\\
\label{eq:31_7}h_ik_i \leq \rho(X+h_iX_i) - \rho(X),\\
\label{eq:31_8}h_ik_i \geq \rho(X) - \rho(X-h_iX_i),
\end{align}
for all $h\in(0,\varepsilon_i)$, $\varepsilon_i$ from (A),
$i=1,\ldots,n$, and for any decomposition $(X_1,\ldots,X_n)$ of
$X=\sum_{i=1}^n X_i$ is suitable for performance measurement with
$\alpha$ at $X$ according to Definition \ref{def:31_1}.
\end{theorem}
\begin{proof}
The proof is similar to the proof of Theorem \ref{prop:suitability}.
However, here we don't rely on the properties ($\overline{S}$) and
(P) of $\theta$. Instead we can use the properties \eqref{eq:31_5}
and \eqref{eq:31_6} that are required from the reward allocation.
With assumption (A) for $\rho$ at $X\in\mathcal{X} \cap
\text{dom\,}\rho \cap \text{dom\,}\theta$ there exists
$\varepsilon_i>0$ such that for all $h_i\in (0,\varepsilon_i)$ we
have $\rho(X)\rho(X+h_i X_i)>0$ and $\rho(X)\rho(X-h_i X_i)>0$ for
any $i=1,\ldots,n$. Let $\theta(X)>0$, then from \eqref{eq:31_6} we
get
\begin{align*}
& \theta(X-h_iX_i)\rho(X)-\theta(X)\rho(X-h_iX_i)\\
\leq
& \theta(X)(\rho(X)-\rho(X-h_iX_i))-h_it_i\rho(X).
\end{align*}
Now \eqref{eq:31_8} gives us
\[
\theta(X)(\rho(X)-\rho(X-h_iX_i))-h_it_i\rho(X)
\leq
h_i(\theta(X)k_i-\rho(X)t_i).
\]
With \eqref{eq:31_1} this leads to $\alpha(X)>\alpha(X-h_iX_i)$.
Similarly we get from \eqref{eq:31_5} and \eqref{eq:31_7} the
inequality
\[
\theta(X+h_iX_i)\rho(X)-\theta(X)\rho(X+h_iX_i)
\leq
h_i(t_i\rho(X)-k_i\theta(X)).
\]
With \eqref{eq:31_3} this leads to $\alpha(X) > \alpha(X+h_iX_i)$.

Now let $\theta(X)<0$. Since $X\in\mathcal{X}$ this means that
$\rho(X)<0$. From this, \eqref{eq:31_5} and \eqref{eq:31_7} we get
\begin{align*}
&\theta(X+h_iX_i)\rho(X)-\theta(X)\rho(X+h_iX_i)\\
\geq
& h_it_i\rho(X)+\theta(X)(\rho(X)-\rho(X+h_iX_i))\\
\geq
& h_i(t_i\rho(X)-k_i\theta(X)).
\end{align*}
Together with \eqref{eq:31_1} this leads to
$\alpha(X+h_iX_i)>\alpha(X)$. By analogous steps we get from
\eqref{eq:31_6} and \eqref{eq:31_8}
\[
\theta(X-h_iX_i)\rho(X)-\theta(X)\rho(X-h_iX_i)
\geq
h_i(k_i\theta(X)-t_i\rho(X))
\]
which leads with \eqref{eq:31_3} to $\alpha(X-h_iX_i)>\alpha(X)$.
\end{proof}

As we have seen in Theorem \ref{thm:31} it was possible to carry over
the suitability for performance measurement approach from the
beginning of Section \ref{sec:suit} to a reward-risk capital
allocation setting. In a similar way as in Corollary \ref{cor:subgr}
we can provide conditions for $\theta$ and $\rho$ that guarantee the
existence of suitable reward-risk capital allocations. The next
corollary summarizes the analogue statements to the results that we
have presented in Corollary \ref{cor:subgr} to Corollary
\ref{cor:gradient-}. These can be immediately recovered from the
previous results.
\begin{corollary}\label{cor:31}
Let $\theta:\mathcal{X}\to\mathbb{R}\cup\{-\infty\}$ be a concave
reward measure and let $\rho:\mathcal{X}\to\mathbb{R}\cup\{+\infty\}$
be a convex risk measure. Then each of the following conditions is
sufficient for the reward-risk allocation
$(t,k)\in\mathbb{R}^n\times\mathbb{R}^n$ to exist and to be suitable
for performance measurement with the reward-risk ratio $\alpha$ at
$X$ according to Definition \ref{def:31_1}.
\begin{itemize}
\item[(a)]

$\rho$ satisfies (A) at $X\in\mathcal{X} \cap \text{dom\,}\rho \cap
\text{dom\,}\theta$, the sets $\partial \theta(X)$ and $\partial \rho(X)$ are nonempty and for any $\xi\in\partial\rho(X)$ and any
$\psi\in\partial\theta(X)$ the reward-risk capital allocation
$(t,k)\in\mathbb{R}^n\times\mathbb{R}^n$ is given by
\begin{equation}\label{eq:31_9}
t_i:=\mathbb{E}[\psi X_i]\quad\text{and}\quad k_i:=\mathbb{E}[\xi X_i]
\end{equation}
for any $i\in\{1,\ldots,n\}$.

\item[(b)]

$X\in\mathcal{X} \cap int(\text{dom\,}\rho) \cap
int(\text{dom\,}\theta)\subset L^p$, $1\leq p <\infty$ and the
reward-risk capital allocation $(t,k)$ is of type \eqref{eq:31_9}
with $\xi\in\partial\rho(X)$ and $\psi\in\partial\theta(X)$.

\item[(c)]

$\theta$ and $\rho$ are finite valued for any
 $X\in  L^p$, $1\leq p < \infty$ and the
 reward-risk capital allocation $(t,k)$ is of type
\eqref{eq:31_9} with $\xi\in\partial\rho(X)$ and
$\psi\in\partial\theta(X)$.

\item[(d)]

$\theta$ and $\rho$ are continuous and G{\^a}teaux-differentiable at the aggregate
position $X\in\mathcal{X}$ and the reward-risk capital allocation
$(t,k)$ is given by
\begin{align*}
k_i&=\lim_{h_i\rightarrow 0} \frac{\rho(X+h_iX_i)-\rho(X)}{h_i},\\
t_i&=\lim_{h_i\rightarrow 0} \frac{\theta(X+h_iX_i)-\theta(X)}{h_i}
\end{align*}
for any $i\in\{1,\ldots,n\}$.
\end{itemize}
\end{corollary}

Condition (d) from Corollary \ref{cor:31} allows us to continue
Example \ref{ex:distortion}.

\begin{example}
In the previous results we have included reward capital allocations to our approach. Thus we can continue Example \ref{ex:distortion}, where we have considered distortion risk measures, to incorporate distortion reward measures as in \eqref{eq:Choquet}. Let the underlying probability space be atomless. Let $\varphi$ be a convex and $\psi$ be a concave distortion function. Let both functions be differentiable. Then $\theta_{\varphi}(X)=\mathbb{E}_{\varphi}[X]$ is a G{\^a}teaux-differentiable coherent reward measure and $\rho_{\psi}(X)=\mathbb{E}_{\psi}[-X]$ is a G{\^a}teaux-differentiable coherent risk measure. With this, according to Definition \ref{def:31_1}, we have to check whether we have
\[
\frac{t_i}{k_i}
=
\frac{\mathbb{E}[X_i\varphi'(1-F_X(X))]}{\mathbb{E}[-X_i\psi'(F_X(X))]}
>
\frac{\mathbb{E}_{\psi}[X]}{\mathbb{E}_{\varphi}[-X]}
\quad\text{or}\quad
\frac{t_i}{k_i}
=
\frac{\mathbb{E}[X_i\varphi'(1-F_X(X))]}{\mathbb{E}[-X_i\psi'(F_X(X))]}
<
\frac{\mathbb{E}_{\psi}[X]}{\mathbb{E}_{\varphi}[-X]}
\]
to make a decision about the investment in subportfolio $X_i$.
\end{example}

G{\^a}teaux-differentiable risk and reward measures allow us to
formulate stronger results regarding suitability for performance
measurement. Thus we will deal with this aspect in a separate
subsection.

\subsection{Suitability for performance mea\-su\-rement with
G{\^a}teaux-dif\-fe\-ren\-tia\-ble reward and risk measures}
\label{subsec:partially_diff_rrr} As in the previous subsection we will
restrict our attention to portfolios that do not allow for arbitrage
opportunities or are irrational. Thus let $\emptyset \neq U\subset
\mathbb{R}^n$ be an open set such that for each $u\in U$ either
\begin{equation}\label{eq:U_1}
\rho_X(u):=\rho\left(\sum_{i=1}^n u_i X_i\right)>0
\quad
\text{and}
\quad
\theta_X(u):=\theta\left(\sum_{i=1}^n u_i X_i\right)>0,
\end{equation}
or
\begin{equation}\label{eq:U_2}
\rho_X(u):=\rho\left(\sum_{i=1}^n u_i X_i\right)<0
\quad
\text{and}
\quad
\theta_X(u):=\theta\left(\sum_{i=1}^n u_i X_i\right)<0.
\end{equation}
Furthermore let $U\subset \mathbb{R}^n$ be such that
$\rho\left(\sum_{i=1}^n u_i X_i\right)\in \mathbb{R}$ and
$\theta\left(\sum_{i=1}^n u_i X_i\right)\in \mathbb{R}$. Analogously
to the previous section we define
\begin{equation}\label{eq:partially_RRR}
\alpha_X(u):=\frac{\theta_X(u)}{\rho_X(u)}
\end{equation}
and denote reward and risk capital allocations, which are vector fields
here, by
\begin{align*}
t&=(t_{1},\ldots,t_{n}):U\rightarrow\mathbb{R}^{n}\\
k&=(k_{1},\ldots,k_{n}):U\rightarrow\mathbb{R}^{n}.
\end{align*}
In this setting, as can be seen from the notation above, each $u\in
U$ represents a portfolio. Here we fix an arbitrary decomposition of
$X$ and the only assumption that will be needed for all the following
results is that $\theta_X$ and $\rho_X$ are partially differentiable
with continuous derivatives in some $u\in U$. This corresponds to
G\^{a}teaux-differentiability of $\theta$ and $\rho$ at some
aggregate position $\sum_{i=1}^n u_i X_i$.

Now, following the lines of \cite{Tasche2004}, we can state a stronger
definition of suitability for performance measurement that fits this
framework.
\begin{definition}
\label{def:Suitable} We deem a reward-risk allocation principle
$(t,k)$ \emph{suitable for performance measurement with} $\alpha_X$
at $u\in U$ if the following two conditions hold:
\begin{enumerate}
\item For any $i\in\left\{ 1,\ldots,n\right\} $ there is some
    $\varepsilon_{i}>0$ such that
\begin{equation}\label{eq:suitable_pos_1}
\frac{t_i(u)}{k_{i}(u)}
>
\frac{\theta_{X}(u)}{\rho_{X}(u)}
\end{equation}
implies
\begin{equation}\label{eq:suitable_pos_2}
\alpha_{X}(u+se_{i})>\alpha_{X}(u)>\alpha_{X}(u-se_{i}).
\end{equation}
for all $s\in(0,\varepsilon_{i})$.
\item For any $i\in\left\{ 1,\ldots,n\right\} $ there is some
    $\varepsilon_{i}>0$ such that
\begin{equation}\label{eq:suitable_neg_1}
\frac{t_i(u)}{k_{i}(u)}
<
\frac{\theta_{X}(u)}{\rho_{X}(u)}
\end{equation}
implies
\begin{equation}\label{eq:suitable_neg_2}
\alpha_{X}(u-se_{i})>\alpha_{X}(u)>\alpha_{X}(u+se_{i}).
\end{equation}
for all $s\in(0,\varepsilon_{i})$.
\end{enumerate}
\end{definition}

Note the difference between our definition of suitability for
performance measurement with $\alpha_X$ and the Definition 3 of
\citet{Tasche2004}. Definition 3 of \citet{Tasche2004} requires the
capital allocation $k$ to satisfy \eqref{eq:suitable_pos_1} and
\eqref{eq:suitable_neg_1} for all linear functions $\theta$ with
$\theta(0)=0$ and all portfolios $u\in U$. In our formulation we
restrict this requirement to a specific reward-risk ratio $\alpha_X$
and a specific portfolio $u\in U$.

Furthermore note the difference between Definition \ref{def:Suitable}
and the Definitions \ref{def:suitability_G} and \ref{def:31_1}.
Definitions \ref{def:suitability_G} and \ref{def:31_1} have a one-sided
nature in \eqref{eq:SUIT_G2} and \eqref{eq:SUIT_G4}, respectively
\eqref{eq:31_2} and \eqref{eq:31_4}, whereas in Definition
\ref{def:Suitable} we are able to formulate \eqref{eq:suitable_pos_2}
and \eqref{eq:suitable_neg_2} as
\begin{align*}
\alpha_{X}(u+se_{i})&>\alpha_{X}(u)>\alpha_{X}(u-se_{i})\\
\alpha_{X}(u-se_{i})&>\alpha_{X}(u)>\alpha_{X}(u+se_{i})
\end{align*}
This is justified by the one-sided nature of the sub-, respectively
superdifferential, of $\theta$ and $\rho$ that was needed to show
existence of suitable capital allocations without the requirement of
G{\^a}teaux-differentiability of $\theta$ and $\rho$. On the other
hand, if we assume that $\theta$ and $\rho$ are
G{\^a}teaux-differentiable then the two-sided nature of the
G{\^a}teaux-derivative allows us to formulate a stronger version of
suitability for performance measurement. This is the formulation of
Definition \ref{def:Suitable}.

Now we can state and prove in a different way our version of Theorem 1
in \cite{Tasche2004}. This version fits the requirements of our
definition of suitability for performance measurement with $\alpha_X$
in $u\in U$ in conjunction with reward-risk ratios instead of mean-risk
ratios. The following result characterizes the gradient risk capital
allocation in this performance measurement framework, if we fix the
reward capital allocation to be the gradient reward capital allocation.

\begin{theorem}\label{thm:grad_alloc_ex_and_uq}
Let $\emptyset\neq U\subset\mathbb{R}^{n}$ be an open set such that
\eqref{eq:U_1} and \eqref{eq:U_2} are true and let
$\theta_{X}:U\rightarrow\mathbb{R}$ and
$\rho_{X}:U\rightarrow\mathbb{R}$ be partially differentiable in $u\in
U$ with continuous derivatives. Furthermore let
$t=(t_{1},\ldots,t_{n}):U\rightarrow\mathbb{R}^{n}$ and
$k=(k_{1},\ldots,k_{n}):U\rightarrow\mathbb{R}^{n}$ be continuous in
$u\in U$. Let $\alpha_X$ be a reward-risk ratio of type
\eqref{eq:partially_RRR}. Then we have the following two statements:
\begin{itemize}
\item[(a)] The gradient reward-risk allocation $(t,k)$, defined by
\begin{align}\label{eq:rw-gradient_allocation}
t_i(u)&:=\frac{\partial \theta_X(u)}{\partial u_i},\\
k_i(u)&:=\frac{\partial \rho_X(u)}{\partial u_i},\label{eq:rs-gradient_allocation}
\quad i=1,\ldots,n,\,  u\in U
\end{align}
is suitable for performance measurement with $\alpha_X$ at $u\in
U$.
\item[(b)]

Let the reward capital allocation
$t=(t_{1},\ldots,t_{n}):U\rightarrow\mathbb{R}^{n}$ be given by the
gradient reward capital allocation from
\eqref{eq:rw-gradient_allocation}. If a risk capital allocation
$k:U\rightarrow \mathbb{R}^n$ is suitable for performance
measurement with $\alpha_X$ in $u\in U$ for any $\theta_X$ that is
partially differentiable in $u\in U$, then it is the gradient risk
capital allocation from \eqref{eq:rs-gradient_allocation}.
\end{itemize}
\begin{proof}
Let us start with (a). Consider the partial derivative of $\alpha_{X}$
in $u_{i}$,
\begin{align}\label{eq:partial_of_alpha}
\frac{\partial\alpha_{X}(u)}{\partial u_{i}}
& =\rho_{X}(u)^{-2}
\left(
\frac{\partial\theta_{X}(u)}{\partial u_{i}}\rho_{X}(u)
-\theta_{X}(u)\frac{\partial\rho_{X}(u)}{\partial u_{i}}
\right)
\end{align}
With \eqref{eq:rw-gradient_allocation} and
\eqref{eq:rs-gradient_allocation}, we get
\[
\frac{\partial\alpha_{X}(u)}{\partial u_{i}}
=\rho_{X}(u)^{-2}
\left(
t_i(u)\rho_{X}(u)
-\theta_{X}(u)k_{i}(u)
\right).
\]
\eqref{eq:suitable_pos_1} now leads to
$\left(\frac{\partial\theta_{X}(u)}{\partial
u_{i}}\rho_{X}(u)-\theta_{X}(u)k_{i}(u)\right)>0$ and therefore it
follows that
\[
\frac{\partial\alpha_{X}(u)}{\partial u_{i}}>0.
\]
This gives us the existence of $\varepsilon_{i}>0$ such that
(\ref{eq:suitable_pos_2}) holds for all $s\in(0,\varepsilon_{i})$,
$i=1,\ldots,n$. The proof of condition 2 in Definition
\ref{def:Suitable} works analogously and the reward-risk capital
allocation from \eqref{eq:rw-gradient_allocation} and
\eqref{eq:rs-gradient_allocation} is suitable for performance
measurement with $\alpha_X$ in $u\in U$.

Let us now prove (b). Fix the reward capital allocation to be the
gradient reward capital allocation. Let $k$ be suitable for performance
measurement for any $\theta_X$ that is partially differentiable in
$u\in U$. We want to show, that $k$ corresponds to
\eqref{eq:rs-gradient_allocation}. Assume without any restriction, that
\begin{equation}\label{eq:assumption1}
k_{i}(u)<\frac{\partial\rho_{X}(u)}{\partial u_{i}},
\quad i\in\{1,\ldots,n\}.
\end{equation}
Since $k$ is suitable for performance measurement with $\alpha_X$ in
$u\in U$ for {\bf any} $\theta_{X}:U\rightarrow\mathbb{R}$ that is
partially differentiable in $u\in U$ we can simply choose
$\theta_{X}^{t}:U\rightarrow\mathbb{R}$ to be
$\theta_{X}^{t}(u)=t\rho_{X}(u)$ for $t>0$. This function is by
assumption partially differentiable in $u\in U$ for each $t>0$. The
choice of $\theta$ provides us now with the following
\begin{align*}
\frac{\partial\theta_{X}^{t}(u)}{\partial u_{i}}
\rho_{X}(u)-k_{i}(u)\theta_{X}^{t}(u)
& =t\frac{\partial\rho_{X}(u)}{\partial u_{i}}
\rho_{X}(u)-k_{i}(u)t\rho_{X}(u)\\
& =t\rho_{X}(u)
\left(
\frac{\partial\rho_{X}(u)}{\partial u_{i}}-k_{i}(u)
\right).
\end{align*}
Now, dependent on the sign of $\rho_X$, from \eqref{eq:assumption1} we
either have
\[
t\rho_{X}(u)\left(\frac{\partial\rho_{X}(u)}{\partial u_{i}}
-k_{i}(u)\right)>0
\]
which should lead to
\[
\alpha_{X}(u+se_{i})>\alpha_{X}(u)
>
\alpha_{X}(u-se_{i}),
\]
or we have
\[
t\rho_{X}(u)\left(\frac{\partial\rho_{X}(u)}{\partial u_{i}}
-k_{i}(u)\right)<0,
\]
which should lead to
\[
\alpha_{X}(u+se_{i})<\alpha_{X}(u)<\alpha_{X}(u-se_{i}),
\]
since $k$ is suitable for performance measurement with $\alpha_X$ in
$u\in U$. But we have $\alpha_{X}(u)
=\frac{\theta_{X}^{t}(u)}{\rho_{X}(u)}
=\frac{t\rho_{X}(u)}{\rho_{X}(u)} =t$ and hence
\[
\alpha_{X}(u+se_{i})=\alpha_{X}(u)
=\alpha_{X}(u-se_{i})\quad\mbox{or}
\quad\frac{\partial\alpha_{X}(u)}{\partial u_{i}}=0,
\]
which is a contradiction to the suitability assumption.
\end{proof}
\end{theorem}

\section{Suitability for performance measurement with
non-divisible\\ portfolios and connections to a game theoretic
context}\label{sec:game}

Let $N=\{1,\ldots,n\}$ and denote by $2^N$ the set of all possible
subsets of $N$. In this section we will work with the functions
\begin{equation}\label{eq:4-1}
\vartheta(S):=\theta\left(\sum_{i\in S}X_i\right)\quad\text{and}\quad
c(S):=\rho\left(\sum_{i\in S}X_i\right)
\end{equation}
where $S$ is a subset of $N$ and the maps $\theta:L^p\to \mathbb{R}$
and $\rho:L^p\to \mathbb{R}$, $1\leq p \leq \infty$, represent, with
the notation from the previous section, the reward and the risk
measure. We have chosen this approach to incorporate a study of
suitability for performance measurement in a context where it is not
possible to further divide the subportfolios $X_i$ of the portfolio
$X=\sum_{i=1}^n X_i$. Here it is only possible to add or to remove
the whole subportfolio $X_i$ to or from the portfolio $X$.

Following the approach in Subsection \ref{subsec:3-1} and
\ref{subsec:partially_diff_rrr} we will focus only on on portfolios
that do not allow for arbitrage gains or are irrational. Thus let
$G\subset 2^N$ be such that for each $S\in G$ we either have
\[
\theta\left(\sum_{i\in S}X_i\right)>0\quad\text{and}\quad
\rho\left(\sum_{i\in S}X_i\right)>0
\]
or
\[
\theta\left(\sum_{i\in S}X_i\right)<0\quad\text{and}\quad
\rho\left(\sum_{i\in S}X_i\right)<0.
\]
Thus the functions $\vartheta$ and $c$ are maps from $G$ to
$\mathbb{R}$ and we will work with the convention
$\vartheta(\emptyset)=c(\emptyset)=0$. In this context the
reward-risk ratio becomes a map $\gamma:G\to (0,\infty)$ which is
given by
\[
\gamma(S)=\frac{\vartheta(S)}{c(S)}
=\frac{\theta\left(\sum_{i\in S}X_i\right)}{\rho\left(\sum_{i\in S}X_i\right)}
\]
and the risk capital allocation is a function $\kappa:G\to
\mathbb{R}^n$ where $\kappa_i(S)$ specifies how much risk capital is
allocated to each subportfolio $X_i$ of the portfolio $\sum_{i\in
S}X_i$. With this notation we can state the appropriate definition of
suitability for performance measurement with $\gamma$ in this context.

\begin{definition}\label{def:game-suitable}
We will call a capital allocation $\kappa$ suitable for performance
measurement with $\gamma$ and a portfolio $\sum_{i\in S}X_i$ of
non-divisible subportfolios if the following two conditions hold:
\begin{enumerate}
\item For any $i\in N\setminus S$
\begin{equation}\label{eq:WS_A-1}
\frac{\theta(X_i)}{\kappa_i(S)}=\frac{\vartheta(\{i\})}{\kappa_i(S)}
>
\frac{\vartheta(S)}{c(S)}=
\frac{\theta\left(\sum_{j\in S}X_j\right)}{\rho\left(\sum_{j\in
S}X_j\right)}
\end{equation}
implies that
\begin{equation}
\frac{\theta\left(\sum_{j\in S\cup\{i\}}X_j\right)}
{\rho\left(\sum_{j\in S\cup\{i\}}X_j\right)}
=\gamma(S\cup\{i\})
>\gamma(S)
=\frac{\theta\left(\sum_{j\in S}X_j\right)}{\rho\left(\sum_{j\in
S}X_j\right)}
.\label{eq:WS_A2-1}
\end{equation}

\item For any $i\in N\setminus S$
\begin{equation}\label{eq:WS_B-1}
\frac{\theta(X_i)}{\kappa_i(S)}=\frac{\vartheta(\{i\})}{\kappa_i(S)}
<
\frac{\vartheta(S)}{c(S)}=
\frac{\theta\left(\sum_{j\in S}X_j\right)}{\rho\left(\sum_{j\in
S}X_j\right)}
\end{equation}
implies that
\begin{equation}
\frac{\theta\left(\sum_{j\in S\cup\{i\}}X_j\right)}
{\rho\left(\sum_{j\in S\cup\{i\}}X_j\right)}
=\gamma(S\cup\{i\})
<\gamma(S)
=\frac{\theta\left(\sum_{j\in S}X_j\right)}{\rho\left(\sum_{j\in
S}X_j\right)}
.\label{WS_B2-1}
\end{equation}
\medskip

\end{enumerate}
\end{definition}

This definition of suitability differs from our previous definitions.
Here, starting from a portfolio $\sum_{i\in S}X_i$, we are interested
in risk allocations that give us the right information whether to add a
specific subportfolio $X_i$, $i\in N\setminus S$, to our portfolio
$\sum_{i\in S}X_i$ or not. If the standalone performance of
subportfolio $X_i$ (measured by $\vartheta(\{i\})/\kappa_i(S)$) is
better than the performance of the whole portfolio $\sum_{i\in S}X_i$
then the subportfolio $X_i$ should be added to the portfolio to improve
the performance of the portfolio. On the other hand, if the standalone
performance of the subportfolio $X_i$, $i\in N\setminus S$, is worse
than the performance of $\sum_{i\in S}X_i$ then adding this
subportfolio to the portfolio should worsen the performance of
$\sum_{i\in S}X_i$.

The connection of this approach to a game theoretic setting can be
established by considering the set $N$ to be a finite set of players
such that the subsets $S\subset N$ can be interpreted as coalitions of
players. Then $c(S)$ represents the cost and $\vartheta(S)$ represents
the gain the coalition of players $S$ generates in a game. From the
viewpoint of the coalition of players $S$ we are interested in cost
allocation principles that give us the right information whether to add
player $i$ to our coalition $S$ or not. In this context $\gamma$
provides a possible tool and Definition \ref{def:game-suitable}
provides the criteria that are needed to make such a decision.

For details about cooperative games we refer the reader to
\cite{Shapley1953}, \cite{Hart1983}, \cite{Faigle1992} and
\cite{Branzei2005}.

In cooperative game theory reward functions like $\vartheta$
(characteristic functions of a game) are often assumed to satisfy
certain properties like superadditivity, which is in this framework the
property
\[
\vartheta(S\cup T)\geq \vartheta(S)+\vartheta(T) \quad
    \text{for all} \quad S,T\subseteq N, \ S\cap
    T=\emptyset.
\]

For the following proposition we will assume that the reward function
$\vartheta$ satisfies a stronger property, namely the additivity
property
\begin{equation}
\vartheta(S\cup\{i\})
=
\vartheta(S)+\vartheta(\{i\})\label{eq:spec_additivity}
\end{equation}
for all $i\in N$ and all $S\subseteq N\setminus\{i\}$. We can see
immediately that this property is equivalent to
\[
\vartheta(S\cup T)
=
\vartheta(S)+\vartheta(T)
\quad
\text{for any } S,T\subseteq N,\, S\cap T=\emptyset.
\]

  This property enables
us to state conditions for allocations to be suitable for performance
measurement with $\gamma$ and a coalition of players $S\in G$ in a
cooperative game-theoretic framework. But we will not need $\vartheta$ to
satisfy property \eqref{eq:spec_additivity} for all $i\in N$ and all
$S\subseteq N\setminus\{i\}$. We will only need the restriction to a
specific coalition of players $S\in G$ and any $i\in N\setminus S$.
\begin{prop}\label{prop:game-suitability}
If the reward function $\vartheta$ satisfies the additivity property
\eqref{eq:spec_additivity} for the coalition of players $S\in G$ and
any $i\in N\setminus S$ then any cost allocation $\kappa$ that
satisfies
\begin{align}
  \kappa_i(S)\geq c(S\cup \{i\})-c(S)
  \quad & \text{if } c(S\cup \{i\})>0
  \quad \text{and}\label{eq:game_1}\\
  \kappa_i(S)\leq c(S\cup \{i\})-c(S)
  \quad & \text{if } c(S\cup \{i\})<0\label{eq:game_2}
\end{align}
for any $i\in N\setminus S$ is suitable for performance measurement
with $\gamma$ and the coalition of players $S\in G$.
\begin{proof}
From \eqref{eq:WS_A-1} together with the additivity property
\eqref{eq:spec_additivity} of $\vartheta$ we get
\begin{align*}
\gamma(S\cup\{i\})-\gamma(S)
& = \frac{\vartheta(S\cup\{i\})c(S)
-\vartheta(S)c(S\cup\{i\})}{c(S\cup\{i\})c(S)}\\
& > \frac{\vartheta(S)[c(S)+\kappa_i(S)-c(S\cup\{i\})]}
{c(S\cup\{i\})c(S)}\\
& = \gamma(S)
\cdot \frac{c(S)+\kappa_i(S)-c(S\cup\{i\})}{c(S\cup\{i\})}.
\end{align*}
With \eqref{eq:game_1} and \eqref{eq:game_2} we get
\[
\gamma(S\cup\{i\})-\gamma(S)>0.
\]
Now, starting with \eqref{eq:WS_B-1} leads with the same arguments to
\[
\gamma(S\cup\{i\})-\gamma(S)<0,
\]
which finally proves the statement.
\end{proof}
\end{prop}

An obvious candidate for a capital allocation that is suitable for
performance measurement with $\gamma$ and the coalition of players $S$
is the marginal contribution allocation defined by
\begin{equation}\label{eq:marginal_allocation}
  \kappa_i(S):=c(S\cup\{i\})-c(S)
=
\rho
\bigg(
\sum_{j\in S}X_j+X_i
\bigg)
-\rho
\bigg(
\sum_{j\in S}X_j
\bigg).
\end{equation}
Thus a straightforward consequence of Proposition
\ref{prop:game-suitability} is the following Corollary.

\begin{corollary}\label{cor:wwa}
If the reward function $\vartheta$ satisfies the additivity property
\eqref{eq:spec_additivity} for the coalition of players $S\in G$ and
any $i\in N\setminus S$  then the mar\-gi\-nal con\-tri\-bution
allocation defined in \eqref{eq:marginal_allocation} is suitable for
performance measurement with $\gamma$ and the coalition of players
$S\in G$.
\end{corollary}

Useful properties a cost allocation in this framework can have are
\begin{description}
\item [{Efficiency:}] $c(N)=\sum_{i=1}^{n}\kappa_i(N)$
\item [{Symmetry:}] $\kappa_i(N)=\kappa_j(N)$ if $c(\{i\}\cup
    S)=c(\{j\}\cup S)$ for all $S\subseteq N\setminus\{i,j\}$
\item [{Dummy~Player~Property:}] $\kappa_i(N)=c(\{i\})$ if
    $c(\{i\}\cup S)=c(\{i\})+c(S)$ for every subset $S\subseteq
    N\setminus\{i\}$.
\end{description}
By restricting the whole framework of this section to the set
$S=N\setminus\{i\}$ the marginal contribution allocation corresponds to
the so called \emph{with-without allocation}, introduced in
\cite{Merton1993} in a risk-capital context, see \eqref{eq:WWA}. The
with-without allocation has the symmetry and the dummy player property,
but in general it fails to have the efficiency property. To achieve
efficiency for this allocation principle a simple normalization
procedure,
\[
k_i
=
\frac{\rho(X)-\rho(X-X_i)}{\sum_{j=1}^n \rho(X)-\rho(X-X_j)}\rho(X),
\]
can be performed, but this is only possible if $\sum_{j=1}^n
\rho(X)-\rho(X-X_j)$ is nonzero.

The formulation of suitability for performance measurement in this
section has a specific advantage. We have shown that there exists an
allocation principle that is suitable for performance measurement with
a performance measure $\gamma$ according to Definition
\ref{def:game-suitable}. This allocation principle is the with-without
allocation known from \cite{Merton1993} and \cite{Matten1996}.
Therefore, if one accepts Definition \ref{def:game-suitable} as
appropriate for a specific problem-setting the advantage of this
allocation, although in general it fails to have the efficiency
property, is that it does not require the cost function (accordingly
the risk measure with the definition in \eqref{eq:4-1}) to have any
specific properties. This allocation principle exists whether the
underlying risk measure is convex, continuous,
G{\^a}teaux-differentiable or not.

The framework of Subsection \ref{subsec:partially_diff_rrr}
corresponds, in contrast to this section, to a fractional players
framework in game theory (see \cite{Aumann1974})  and thus the
results obtained in Subsection \ref{subsec:partially_diff_rrr} are
easily transferable to a game theoretic setting by using  a
fractional cost and reward function, ie,
\begin{align*}
\vartheta(u)=\theta\left(\sum_{i=1}^n u_i X_i \right),
\quad & u\in U\subset\mathbb{R}^n\\
c(u) = \rho\left(\sum_{i=1}^n u_i X_i \right),
\quad &  u\in U\subset\mathbb{R}^n.
\end{align*}

\bibliography{suitability}
\bibliographystyle{chicago}

\end{document}